\documentclass[11pt]{article}

\usepackage{fullpage}

\usepackage{amsthm,amsmath}  
\usepackage{xspace,enumerate}
\usepackage[utf8]{inputenc}
\usepackage{thmtools}
\usepackage{thm-restate}
\usepackage{authblk}
\usepackage[hypertexnames=false,colorlinks=true,urlcolor=Blue,citecolor=Green,linkcolor=BrickRed]{hyperref}
\usepackage[dvipsnames,usenames]{color}
\usepackage{graphicx}
\usepackage[noadjust]{cite}
\usepackage{microtype}
\usepackage{todonotes}
\usepackage{thmtools}
\usepackage{thm-restate}
\usepackage[capitalise]{cleveref}
\usepackage{algorithmic}
\usepackage[pagewise]{lineno}
\nolinenumbers

\theoremstyle{plain}
\newtheorem{theorem}{Theorem}
\newtheorem{lemma}[theorem]{Lemma} 
  
\newtheorem{proposition}[theorem]{Proposition}

\newtheorem{definition}[theorem]{Definition}

\author[1]{Pawe\l{} Gawrychowski}
\author[2]{Shay Mozes}
\author[3]{Oren Weimann}

\affil[1]{
University of Wroc\l{}aw, Poland\\
\href{mailto:gawry@cs.uni.wroc.pl}{gawry@cs.uni.wroc.pl}}

\affil[2]{
The Interdisciplinary Center Herzliya, Israel\\
\href{mailto:smozes@idc.ac.il}{smozes@idc.ac.il}}

\affil[3]{
University of Haifa, Israel\\
\href{mailto:oren@cs.haifa.ac.il}{oren@cs.haifa.ac.il}
}

\date{}

\title{Minimum Cut in $O(m\log^2 n)$ Time}

\begin{document}

\maketitle

\thispagestyle{empty}

\begin{abstract}
We give a randomized algorithm that finds a minimum cut in an undirected weighted $m$-edge $n$-vertex graph $G$ with high probability in $O(m \log^2 n)$ time. This is the first improvement to Karger's celebrated $O(m \log^3 n)$ time algorithm from 1996. Our main technical contribution is a deterministic $O(m \log n)$ time algorithm that, given a spanning tree $T$ of $G$, finds a minimum cut of $G$ that 2-respects (cuts two edges of) $T$. 
\end{abstract}

\clearpage
\setcounter{page}{1}

\section{Introduction}
The minimum cut problem is one of the most fundamental and well-studied optimization problems in theoretical computer science. Given an undirected edge-weighted graph $G=(V,E)$, the problem asks to find a subset of vertices $S$ such that the total weight of all edges between $S$ and $V \setminus S$ is minimized. The vast literature on the minimum cut problem can be classified into three main approaches:

\medskip
\noindent
 {\bf The maximum-flow approach.}
The minimum cut problem was originally solved by computing the maximum $st$-flow~\cite{FF62} for all pairs of vertices $s$ and $t$. In 1961, Gomory and Hu~\cite{GomoryHu} showed that only $O(n)$ maximum $st$-flow computations are required, and in 1994 Hao and Orlin~\cite{HaoO94} showed that in fact a single maximum $st$-flow computation suffices. A maximum $st$-flow can be found in $O(mn\log(n^2/m))$ time using the Goldberg-Tarjan algorithm~\cite{GoldbergTarjan}, and the fastest algorithm to date takes $O(mn)$ time~\cite{Orlin,KingRT94}. 
Faster (though not near-linear time) algorithms are known (see e.g~\cite{GoldbergRao,LeeSidford,Madry16} and references within) when the graph is unweighted or when the maximum edge weight $W$ is not extremely large. 

\medskip
\noindent
{\bf The edge-contraction approach.}
An alternative method 
 is edge contraction.  If we can identify an edge that does not cross the minimum cut, then we can contract this edge without affecting the minimum cut. Nagamochi and Ibaraki~\cite{NI92a,NI92b} showed how to deterministically find a contractible edge in $O(m)$ time, leading to an $O(mn+n^2\log n)$-time minimum cut algorithm. 
Karger~\cite{Karger93} showed that randomly choosing the edge to contract works well with high probability. In particular, Karger and Stein~\cite{KargerStein} showed that this leads to an improved $O(n^2 \log^3 n)$ Monte Carlo algorithm.

\medskip
\noindent
{\bf The tree-packing approach.}
In 1961, Nash-Williams~\cite{NW61} proved that, in unweighted graphs, any graph with minimum cut $c$ contains a set of $c/2$ edge-disjoint spanning trees. Gabow's algorithm~\cite{Gabow} can be used to find such a tree-packing with $c/2$ trees in $O(mc\log n)$ time.
Karger~\cite{Karger} observed that the $c$ edges of a minimum cut must be partitioned among
these $c/2$ spanning trees, hence the minimum cut {\em 1-} or {\em 2-respects} some tree in the packing.
That is, one of the trees is such that at most two of its edges cross the minimum cut (these edges are said to {\em determine} the cut).
We can therefore find the minimum cut by examining each tree and finding the minimum cut that 1- or 2-respects it. 

Several obstacles need be overcome in order to translate this idea into an efficient minimum cut algorithm for weighted graphs:
(1) we need a weighted version of tree-packing,
(2) finding the packing (even in unweighted graphs) takes time proportional to $c$ (and $c$ may be large),
(3) checking all trees takes time proportional to $c$, and
(4) one needs an efficient algorithm that, given a spanning tree $T$ of $G$, finds the minimum cut in $G$ that 2-respects $T$ (finding a minimum cut that 1-respects $T$ can be easily done in $O(m+n)$ time, see e.g~\cite[Lemma 5.1]{Karger}).

In a seminal work, Karger~\cite{Karger} overcame all four obstacles: First, he converts $G$ into an unweighted graph by conceptually replacing an edge of weight $w$ by $w$ parallel edges. Then, he uses his random sampling from~\cite{Kar97a,Karger93} combined with Gabow's algorithm~\cite{Gabow} to reduce the packing time to $O(m + n \log^3 n)$ and the number of trees in the packing to $O(\log n)$. Finally, he designs a deterministic $O(m \log^2 n)$ time algorithm that given a spanning tree $T$ of $G$ finds the minimum cut in $G$ that 2-respects $T$. Together, this gives an $O(m \log^3 n)$ time randomized algorithm for minimum cut. Until the present work, this was the fastest known algorithm for undirected weighted graphs. 

Karger's $O(m \log^2 n)$ algorithm for the 2-respecting problem finds, for each edge $e\in T$, the edge $e'\in T$ that minimizes the cut determined by $\{e,e'\}$. He used link-cut trees~\cite{LinkCutTree} to efficiently keep track of the sizes of cuts as the candidate edges $e$ of $T$ are processed in a certain order (bough decomposition), consisting of $O(\log n)$ iterations, and guarantees that the number of dynamic tree operations is $O(m)$ per iteration. Since each link-cut tree operation takes $O(\log n)$ time, the total running time for solving the 2-respecting problem is $O(m \log^2 n)$.
  
In a very recent paper, Lovett and Sandlund~\cite{SimpleArxiv} proposed a simplified version of Karger's algorithm. Their algorithm  has the same $O(m\log^3 n)$ running time as Karger's. To solve the 2-respecting problem they use top trees~\cite{TopTrees} rather than link-cut trees, and 
use heavy path decomposition~\cite{HT1984,LinkCutTree} to guide the order in which edges of $T$ are processed. A property of heavy path decomposition is that, for every edge $(u,v) \notin T$, the $u$-to-$v$ path in $T$ intersects $O(\log n)$ paths of the decomposition. This property implies that the contribution of each non-tree edge to the cut changes $O(\log n)$ times along the the entire process. See also~\cite{GeissmannG18} who use the fact that the bough decomposition, implicitly used by Karger, also satisfies the above property. 
The idea of traversing a tree according to a heavy path decomposition, i.e., by first processing a smaller subtree and then processing the larger subtree has been used quite a few times in similar problems on trees. See e.g.,~\cite{BrodalFP01}.
While the ideas of Lovett and Sandlund~\cite{SimpleArxiv} do not improve on Karger's bound, their paper has drawn our attention to this problem. 

\subsection{Our result and techniques} In this paper, we present a deterministic $O(m \log n)$ time algorithm that, given a spanning tree $T$ of $G$, finds the minimum cut in $G$ that 2-respects $T$. Using Karger's framework, this implies an $O(m \log^2 n)$ time randomized algorithm for minimum cut in weighted graphs.

Like prior algorithms for this problem, our algorithm finds, for each edge $e\in T$ the edge $e' \in T$ that minimizes the cut determined by $\{e,e'\}$. The difficult case to handle is when $e$ and $e'$ are such that neither of them is an ancestor of the other. In Karger's solution, handling each edge $e=(u,v)\in T$ was done using amortized $O(d\log n)$ operations on Sleator and Tarjan's link-cut trees~\cite{LinkCutTree} where $d$ is the number of non-tree edges incident to $u$. Since operations on link-cut trees require $O(\log n)$ amortized time, the  time to handle all edges is $O(m\log^2 n)$ (implying an $O(m \log^3 n)$ time algorithm for the minimum cut problem). 
As an open problem, Karger~\cite{Karger} asked (more than 20 years ago) whether the required link-cut tree operations can be done in constant amortized time per operation (implying an $O(m \log^2 n)$ time algorithm for the minimum cut problem). Karger even pointed out why one could perhaps achieve this: ``{\em We are not using the full power of dynamic trees (in particular, the tree we are operating on is static, and the sequence of operations is known in advance)}.'' In this paper, we manage to achieve exactly that. 
We show how to order the link cut tree operations so that they can be handled efficiently in batches. 
We call such a batch a {\em bipartite problem} (see Definition~\ref{def:bipartite}). 

Perhaps a reason that the running time of Karger's algorithm has not been improved in more than two decades is that it is not at all apparent that these bipartite problems can indeed be solved more efficiently. Coming up with an efficient solution to the bipartite problem requires a combination of several additional ideas. 
Like~\cite{SimpleArxiv}, we use heavy path decomposition, but in a different way. We develop a new decomposition of a tree that combines heavy path decomposition with biased divide and conquer, and use this decomposition in conjunction with a compact representation which we call topologically induced subtrees (see Definition~\ref{def:induced}). This compact representation turns out to be crucial not only for solving the bipartite problem, but also to the reduction from the original problem to a collection of bipartite problems.

\subsection{Application to unweighted graphs}
Karger's method is inherently randomized and obtaining a deterministic (or at least Las Vegas) near-linear time algorithm for the minimum cut in undirected weighted graphs is an interesting open problem. For unweighted graphs, such a deterministic algorithm was provided by Kawarabayashi and Thorup~\cite{KawarabayashiT19}. Later, Henzinger, Rao, and Wang~\cite{HenzingerRW17} designed a faster $O(m\log^{2} n(\log\log n)^{2})$ time algorithm. Very recently Ghaffari, Nowicki and Thorup~\cite{GhaffariNT19} introduced a new technique of random 2-out contractions and applied it to design an $O(\min\{m+n\log^{3}n,m\log n\})$ time randomized algorithm that finds a minimum cut with high probability. We stress that the faster algorithms of Henzinger et al. and Ghaffari et al. work only for unweighted graphs, that is, for edge connectivity. Interestingly, the latter uses Karger's $O(m\log^{3}n)$ time algorithm as a black box, and by plugging in our faster method one immediately obtains an improved running time of $O(\min\{m+n\log^{2}n,m\log n\})$ for unweighted graphs.

\subsection{Independent work}
Independently to our work\footnote{To be accurate, their work appeared on arXiv one day after ours.}, Mukhopadhyay and Nanongkai~\cite{Danupon} came up with an $O(m\log n + n\log^4 n)$ time algorithm for finding a minimum 2-respecting cut.
While this improves Karger's bound for sufficiently dense graphs, it does not improve it for all graphs, and is randomized.
Our algorithm uses a different (deterministic and simple) approach and strictly dominates both Karger's and Mukhopadhyay and Nanongkai's running time for all graphs.
There are however benefits to the approach of~\cite{Danupon} in other settings.
Namely, they use it to obtain an algorithm that requires $\tilde O(n)$ cut queries to compute the min-cut, and a streaming
algorithm that requires $\tilde O(n)$ space and $O(\log n)$ passes to compute the min-cut.

\section{Preliminaries}

\subsection{Karger's algorithm}

At a high level, Karger's algorithm~\cite{Karger} has two main steps. The input is a weighted undirected graph $G$. The first step produces a set $\{T_1,\dots, T_s\}$ of $s=O(\log n)$ spanning trees such that, with high probability, the minimum cut of $G$ $1$- or $2$-respects at least one of them. The second step deterministically computes for each $T_i$ the minimum cut in $G$ that 1-respects $T_i$ and the minimum cut in $G$ that 2-respects $T_i$. The minimum cut in $G$ is the minimum among the cuts found in the second step.

Karger shows~\cite[Theorem~4.1]{Karger} that producing the trees  $\{T_1,\dots, T_s\}$ in the first step can be done in $O(m+n\log^3 n)$ time, and that finding the minimum 2-respecting cut for all the $T_i$'s in the second step can be done in $O(m\log^3 n)$ time.  We will show that each of the steps can be implemented in $O(m \log^2 n)$ time. Showing this for the second step is the main result of the paper, and is presented in Section~\ref{sec:main}. For the first step, we essentially use Karger's proof. Since the first step was not the bottleneck in Karger's paper, proving a bound of $O(m+n\log^3 n)$ was sufficient for his purposes. Karger's concluding remarks suggest that he knew that the first step could be implemented in $O(m\log^2 n)$ time. For completeness, we prove the $O(m\log^2 n)$ bound by slightly modifying Karger's arguments and addressing a few issues that were not important in his proof. Readers proficient with Karger's algorithm can safely skip the proof.

\begin{definition}[2-respecting and 2-constraining]
	Given a spanning tree $T$ and a cut $(S,\bar S)$, we say that the cut 2-respects $T$ and that $T$ 2-constrains the cut if at most 2 edges of $T$ cross the cut. 
\end{definition}

\begin{definition}[weighted tree packing]
Let $G$ be an unweighted undirected graph. Let $\mathcal T$ be a set of spanning trees of $G$, where each tree $T\in \mathcal T$ is assigned a weight $w(T)$. We say that the \emph{load} of an edge $e$ of $G$ (w.r.t. $\mathcal T$) is $\ell(e) = \sum_{T \in \mathcal T : e \in T} w(T)$. We say that $\mathcal T$ is a \emph{weighted tree packing} if no edge has load exceeding 1. The \emph{weight} of the packing $\mathcal T$ is $\tau = \sum_{T \in \mathcal T} w(T)$.	
\end{definition}

\begin{theorem}
	Given a weighted undirected graph $G$, in $O(m \log^2 n)$ time, we can construct a set $\mathcal T$ of $O(\log n)$ spanning trees such that, with high probability, the minimum cut $2$-respects at least one of the trees in $\mathcal T$.
\end{theorem}
\begin{proof}
Let $(S,\bar S)$ be the partition of the vertices of $G$ that forms a minimum cut, and let $c$ be the weight of the minimum cut in $G$. 
We precompute a constant factor approximation of $c$ in $O(m \log^2 n)$ time using Matula's algorithm~\cite{Matula}. See Appendix~\ref{app:Matula} for details.

We assume that all weights are integers, each fitting in a single memory word.
Since edges with weight greater than $c$ never cross the minimum cut, we contract all edges with weight greater than our estimate for $c$, so that now the total weight of edges of $G$ is $O(mc)$.

For the sake of presentation we think of an unweighted graph $\tilde G$, obtained from $G$ by replacing an edge of weight $w$ by $w$ parallel edges. We stress that $\tilde G$ is never actually constructed by the algorithm.
Let $\tilde m$ denote the number of edges of $\tilde G$. By the argument above, $\tilde m = O(mc)$.  
Let $p=\Theta(\log n/c)$. 
Let $H$ be the unweighted multigraph obtained by sampling $\lceil p\tilde m \rceil$ edges of $\tilde G$ ($H=\tilde G$ if $c<\log n$). Clearly, the expected value of every cut in $H$ is $p$ times the value of the same cut in $G$.
By~\cite[Lemma 5.1]{Kar97a}, choosing the appropriate constants in the sampling probability $p$ guarantees that, with high probability, the value of every cut in $H$ is at least $64/65$ times its expected value, and no more than $66/65$ times its expected value. It follows that, with high probability, (i) the minimum cut in $H$ has value $c' = \Theta(\log n)$, and that (ii) the value of the cut in $H$ defined by $(S,\bar S)$ is at most $33c'/32$.

The conceptual process for constructing $H$ can be carried out by randomly selecting $\lceil p\tilde m \rceil$ edges of $G$ (with replacement) with probability proportional to their weights. Since the total edge weight of $G$ is $O(mc)$, each selection can be easily performed in $O(\log(mc))$ time. Using a standard technique~\cite{KY76}, each selection can actually be done with high probability in $O(\log m)$ time. Thus, the time to construct $H$ is $O(p\tilde m \log m) = O(m\log^2 n)$.
We emphasize that $H$ is an unweighted multigraph with $m' = O(m \log n)$ edges, and note that we can assume no edge of $H$ has multiplicity greater than $c'$ (we can just delete extra copies).

Next, we apply the following specialized instantiation~\cite[Theorem 2]{Thorup00} of Young's variant~\cite{Young95} of the Lagrangian packing technique of Plotkin, Shmoys, and Tardos~\cite{PlotkinST95}. It is shown~\cite{Thorup00,Young95} that  for an unweighted graph $H$ with $m'$ edges and minimum cut of size $c'$, the following algorithm finds a weighted tree packing  of weight $3c'/8 \leq \tau \leq c'$. 

\begin{algorithmic}[1]
\STATE $\ell(e) := 0$ for all $e\in E(H)$
\WHILE{there is no $e$ with $\ell(e) \geq 1$}
   \STATE find a minimum spanning tree $T$ w.r.t. $\ell(\cdot)$
   \STATE $w(T) = w(T) + 1/(96\ln m')$
   \STATE $\ell(e)  = \ell(e) + 1/(96\ln m')$ for all $e\in T$
\ENDWHILE	
\end{algorithmic}

Karger~\cite[Lemma 2.3]{Karger} proves that for a graph $H$ with minimum cut $c'$, any tree packing of weight at least $3c'/8$, and any cut $(S,\bar S)$ of $H$ of value at most $33c'/32$, at least a $1/8$ fraction of the trees (by weight) 2-constrain the cut $(S,\bar S)$. Thus, a tree chosen at random from the packing according to the weights 2-constrains the cut $(S,\bar S)$ with probability at least 1/8. Choosing $O(\log n)$ trees guarantees that, with high probability, one of them 2-constrains the cut $(S,\bar S)$, which is the minimum cut in $G$.

It remains to bound the running time of the packing algorithm. Observe 
 that the algorithm increases the weight of some tree by $1/(96\ln m')$ at each iteration. Since the weight $\tau$ of the resulting packing is bounded by $c'$, there are at most $96 c' \ln m' = O(\log^2 n)$ iterations. The bottleneck in each iteration is the time to compute a minimum spanning tree in $H$. We argue that this can be done in $O(m)$ time even though $m' = O(n\log n)$. To see this, first note that since $H$ is a subgraph of $G$, $H$ has at most $m$ edges (ignoring multiplicities of parallel edges). Next note that it suffices to invoke the MST algorithm on a subgraph of $H$ that includes just the edge with minimum load among any set of parallel edges. Since the algorithm always increases the load of edges by a fixed  amount, the edge with minimum load in each set of parallel edges can be easily maintained in $O(1)$ time per load increase by maintaining a cyclic  ordered list of each set of parallel edges and moving to choose the next element in  this cyclic list whenever the load of the current element is incremented. Hence, we can invoke the randomized linear time MST algorithm~\cite{KargerKT95} on a simple subgraph of $H$ of size $O(m)$. 
It follows that the running time of the packing algorithm, and hence of the entire procedure, is $O(m\log^2 n)$.
\end{proof}  

\subsection{Link-cut trees}

In our algorithm we will repeatedly use a structure that maintains a rooted tree $T$ with costs on the edges under the following operations:
\begin{enumerate}
\item $T.\textsc{add}(u,\Delta)$ adds $\Delta$ to the cost of every edge on the path from $u$ to the root,
\item $T.\textsc{path}(u)$ finds the minimum cost of an edge on the path from $u$ to the root,
\item $T.\textsc{subtree}(u)$ finds the minimum cost of an edge in the subtree rooted at $u$.
\end{enumerate}
All three operations can be supported with a link-cut tree~\cite{LinkCutTree} in amortized $O(\log |T|)$ time.\footnote{The original paper~\cite{LinkCutTree} did not include the third operation. However, as shown in~\cite[Appendix 17]{PlanarBook}, it is not difficult to add it.}
We note that we only require these three operations and do not actually use the link and cut functionality of link-cut trees. Other data structures might also be suitable. See, e.g., the use of top-trees in~\cite{SimpleArxiv}.

\subsection{Topologically induced subtrees}

For a rooted tree $T$ and a node $v$ we denote by $T_v$ the subtree of $T$ rooted at $v$. For an edge $e$ of $T$ we denote by $T_e$ the subtree of $T$ rooted at the lower endpoint of $e$.

Let $T$ be a binary tree with edge-costs and $n$ nodes, equipped with a data structure that can answer lowest common ancestor (LCA) queries on $T$ in constant time~\cite{HT1984}. Let $\Lambda = \{w_1,w_2,\ldots,w_s\}$ be a subset of nodes of $T$. We define a smaller tree $T^\Lambda$ that is equivalent to $T$ in the following sense: 

\begin{definition}[topologically induced tree]\label{def:induced}
We say that a tree $T^\Lambda$ is topologically induced on $T$ by $\Lambda$ if for every $S \subseteq \Lambda$, the minimum cost edge $f\in T^\Lambda$ with $T^\Lambda_f \cap \Lambda = S$ has the same cost as the  minimum cost edge $e\in T$ with $T_e \cap \Lambda = S$. 	
\end{definition}
To be clear, the above definition implies that, for any $S \subseteq \Lambda$, there is an edge $e \in T$ with $T_e \cap \Lambda = S$, if and only if there is an edge $f \in T^\Lambda$ with $T^\Lambda_f \cap \Lambda = S$.   
The term \emph{topologically induced tree} will be justified by the construction in the following lemma.

\begin{figure}[h]
\begin{center}
\includegraphics[scale=0.98]{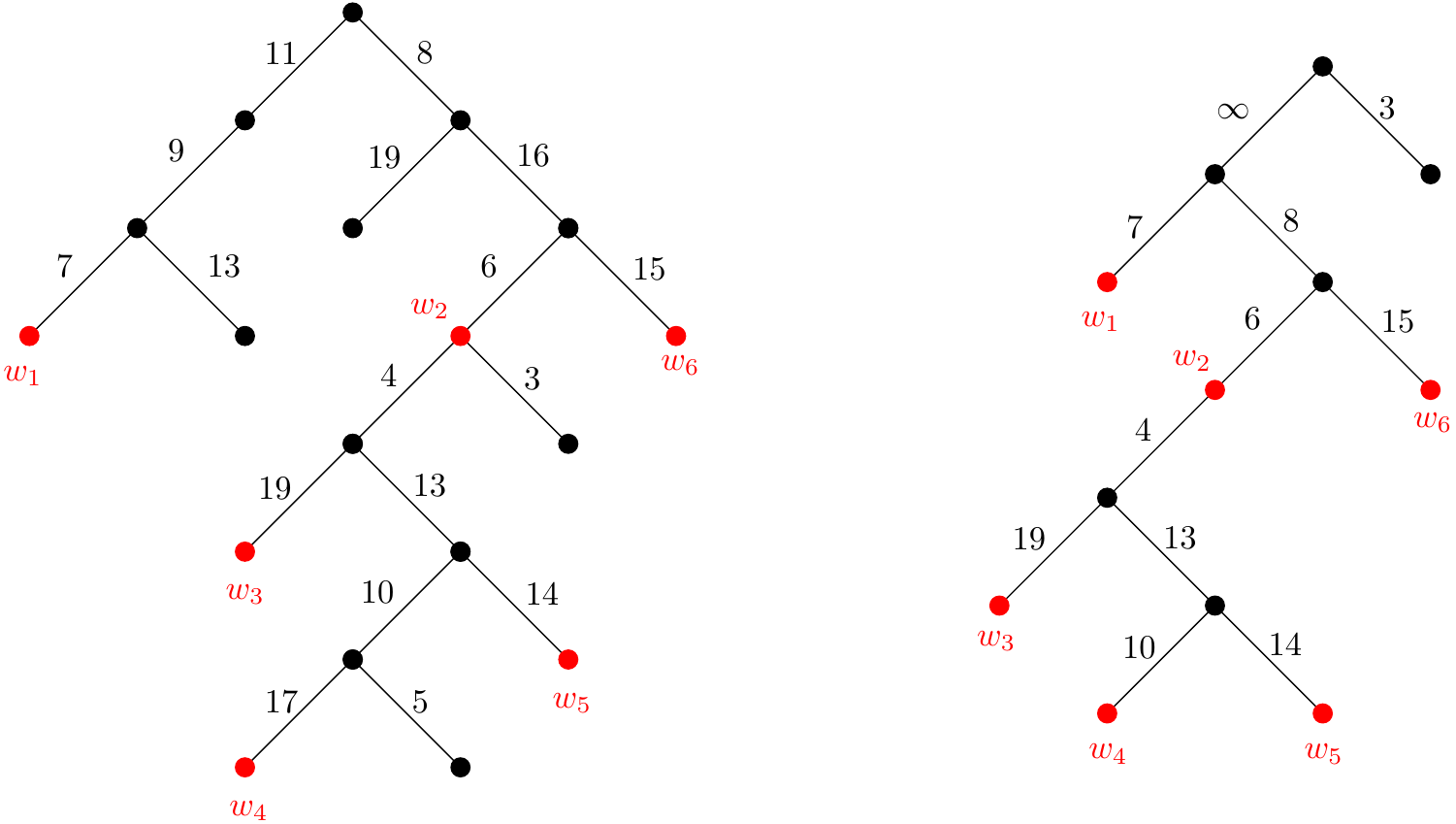}
\end{center}
\caption{On the left: a tree and (in red) a set $\Lambda=\{w_{1},w_{2},\ldots,w_{6}\}$ sorted according to their preorder numbers.
On the right: the corresponding topologically induced tree.}
\label{fig:induced}
\end{figure}

\begin{lemma}
	\label{lem:topology}
	There exists an algorithm that, given a binary tree $T$ with edge costs, equipped with a link-cut data structure, and a list $\Lambda = \{w_1,w_2,\ldots,w_s\}$ of nodes of $T$, ordered according to their visit time in a preorder traversal of $T$, 
constructs in $O(\min\{|T|, s\log|T|\})$ time, a tree $T^\Lambda$ of size $O(s)$ that is topologically induced on $T$ by $\Lambda$.
\end{lemma}
\begin{proof}
We define the tree $T^\Lambda$ to be a tree over all nodes $w_i \in \Lambda$, together with the root and the lowest common ancestor in $T$ of every pair of nodes $w_i$ and $w_j$. For any two nodes $u,v\in T^\Lambda$, $u$ is an ancestor of $v$ in $T^\Lambda$ if and only if $u$ is an ancestor of $v$ in $T$. Thus, each edge $(u,v)$ of $T^\Lambda$ corresponds to the $u$-to-$v$ path in $T$. The edges on this path in $T$ are exactly the edges $e$ of $T$ with $T_e 
\cap \Lambda = T^\Lambda_{(u,v)} \cap \Lambda$. Hence, the paths of $T$ corresponding to distinct edges of $T^\Lambda$ are edge disjoint. We define the cost of the edge $(u,v)$ of $T^\Lambda$ to be the minimum cost of an edge on the corresponding path in $T$. It follows that for every $\emptyset \neq S \subseteq \Lambda$, the minimum cost edge $f\in T^\Lambda$ with $T^\Lambda_f \cap \Lambda = S$ has the same cost as the  minimum cost edge $e\in T$ with $T_e \cap \Lambda = S$.
To guarantee that this condition holds for $S = \emptyset$ as well, we choose the edge $e$ of $T$ with the minimum 
cost such that $T_e \cap \Lambda = \emptyset$ and proceed as follows if such an edge exists.
We create a new node $v$ and change the root of $T^\Lambda$ to $v$ by making the old root of $T^{\Lambda}$ a child of $v$
via an edge with infinite cost. We then add a new edge incident to $v$, whose cost is set to the cost of $e$.
This transformation does not change $T_e \cap \Lambda$ for any edge $e$ of $T^\Lambda$, but now that condition with $S=\emptyset$ is satisfied for the new edge incident to the root.  

We now turn to proving the construction time. We first prove that $T^\Lambda$ consists of at most $2s$ nodes. This is because $T^\Lambda$ consists only of the root, the nodes $w_i$ and LCA$(w_i,w_{i+1})$. To see this, consider two nodes $w_i$ and $w_j$ with $i<j$ such that their lowest common ancestor $u$ is different than $w_i$ and $w_j$. Let $u_\ell$ ($u_r$) be the left (right) child of $u$. Then, $w_i$ is a descendant of $u_\ell$ and $w_j$ a descendant of $u_r$. Let $i'$ be the largest index  such that $w_{i'}$ is in the subtree rooted at $u_\ell$. Then $u=$LCA$(w_{i'},w_{i'+1})$.

We next prove that $T^\Lambda$ can be constructed in $O(s)$ time. We use a method similar to constructing
the Cartesian tree~\cite{Cartesian1,Cartesian2} of a sequence: we scan $w_1,w_2,\ldots,w_s$ from the left to right while maintaining the
subtree of $T^\Lambda$ induced by $w_0 = $ LCA$(w_1,w_s)$, and $w_1,w_2,\ldots,w_i$. initially, the subtree of $T^\Lambda$ induced by $w_0$ and $w_1$ is just a single edge $(w_0,w_1)$. We keep the rightmost path of the subtree of $T^\Lambda$ induced by $w_0, w_1, \dots, w_i$ 
on a stack, with the bottommost edge on the top. To process $w_{i+1}$, we first
find $x = $ LCA$(w_i,w_{i+1})$. Then, we pop from the stack all edges 
$(u,v)$ such that $u$ and $v$ are both below (or equal to) $x$ in $T$. Finally, we possibly split
the edge on the top of the stack into two and push a new edge onto the stack. The amortized
complexity of every step is constant, so the total time is $O(s)$.

Once $T^\Lambda$ is constructed, we set the cost of every edge $(u,v)$ in $T^\Lambda$ to be the minimum cost of an edge
on the $u$-to-$v$ path in $T$. This can be done in $O(\log |T|)$ time per edge of $T^\Lambda$ by
first calling $T.\textsc{add}(u)(\infty)$, then $T.\textsc{path}(v)$ to retrieve the answer, and finally $T.\textsc{add}(u)(-\infty)$,
for a total of $O(s \log|T|)$ time.
Alternatively, we can explicitly go over the edges of the corresponding paths of $T$ for every edge of $T^\Lambda$.
We had argued above that these paths are disjoint so this takes $O(|T|)$ in total.

We also need to compute the cost of the edge $e$ of $T$ with minimum cost such that $T_e \cap \Lambda = \emptyset$. To this end, for each $v \in \Lambda$ we add $\infty$ to the cost of all edges on the path from $v$ to the root of $T$. This takes $O(\min(|T|,s \log |T|)$ by either a bottom up computation on $T$, or using $T.\textsc{add}(u,\infty)$ for every $v\in\Lambda$. We then retrieve the edge with minimum cost in the entire tree in $O(\log|T|)$ time by a call to $\textsc{subtree}$ for the root of $T$, and then subtract $\infty$ from the cost of all edges on the path from $v$ to the root of $T$ for every $v\in\Lambda$. 
\end{proof}

We will use the fact that the operation of taking the topologically induced subtree is composable in the following sense.
\begin{proposition}
	\label{prop:topo-compose}
Let $T$ be a binary tree with edge-costs. Let $\Lambda_2 \subseteq \Lambda_1$ be subsets of nodes of $T$. Let $T_1$ 
be topologically induced on $T$ by $\Lambda_{1}$ and $T_{2}$ be topologically induced on $T_{1}$ by $\Lambda_{2}$.
Then $T_{2}$ is topologically induced on $T_{1}$ by $\Lambda_{2}$.
\end{proposition}

\section{Finding a Minimum 2-respecting Cut}\label{sec:main}
Given a graph $G$ and a spanning tree $T$ of $G$, a cut in $G$ is said to {\em 2-respect} the tree $T$ if at most two edges $e,e'$ of $T$ cross the cut (these edges are said to {\em determine} the cut). In this section we prove the main theorem of this paper:

\begin{theorem}\label{thm:2respecting}
Given an edge-weighted graph $G$ with $n$ vertices and $m$ edges and a spanning tree $T$, the minimum (weighted) cut in $G$ that 2-respects $T$ can be found in $O(m\log n)$ time. 
\end{theorem}

The minimum cut determined by every single edge can be easily found in $O(m+n)$ time~\cite[Lemma 5.1]{Karger}. We therefore focus on finding the minimum cut determined by two edges. 
 Observe that the cut determined by $\{e,e'\}$ is unique and consists of all edges $(u,v) \in G$ such that the  $u$-to-$v$ path in $T$ contains exactly one of $\{e,e'\}$.

We begin by transforming $T$ (in linear time) into a binary tree. This is standard and is done by replacing every node of degree $d$ with a binary tree of size $O(d)$ where internal edges have weight $\infty$ and edges incident to leaves have their original weight. We also add an artificial root to $T$ and connect it to the original root with an edge of weight $\infty$. From now we will be working with binary trees only.
  
\subsection{Descendant edges}\label{sec:descendant}
 We first describe an $O(m\log n)$ time algorithm for finding the minimum cut determined by all pairs of edges $\{e,e'\}$ where $e'$ is a descendant of $e$ in $T$ (i.e. $e'$ is in the subtree of $T$ rooted at the lower endpoint of $e$). 
To this end we shall efficiently find, for each edge $e$ of $T$, the descendant edge $e'$ that minimizes the weight of the cut determined by $\{e,e'\}$, and return the pair minimizing the weight of the cut.
 
For a given edge $e$ of $T$, let $T_e$ denote the subtree of $T$ rooted at the lower endpoint of $e$. We associate with every node $x$ a list of all edges $(u,v)$ such that $x$ is the lowest common ancestor of $u$ and $v$. Note that all these lists can be computed in linear time, and form a partition of the edges of $G$.
We also compute in $O(m)$ time, for every edge $e$ of $T$, the total weight $A(e)$ of all edges with exactly one endpoint in $T_e$ (in fact, this is the information computed by Karger's algorithm for the 1-respecting case). Note that $A(e)$ includes the weight of $e$. 

Using a link-cut tree we maintain a {\em score} for every edge $e$ of $T$. All scores are first initialized to zero. Then, for every edge $(u,v)$ of $G$, we increase the score of all edges on the $u$-to-$v$ path in $T$ by the weight $w(u,v)$ of $(u,v)$. 
This takes $O(\log n)$ time per edge $(u,v)$ by calling $T.\textsc{add}(u,w(u,v))$, $T.\textsc{add}(v,(u,w(u,v))$ and $T.\textsc{add}(\text{LCA}(u,v),-2(u,w(u,v))$.
This initialization takes $O(m \log n)$ time. We then perform an Euler tour of $T$. When the tour first descends below a node $x$, for every edge $(u,v)$ in the list of $x$, we decrease the score of all edges on the $u$-to-$v$ path in $T$ by $2w(u,v)$.  Note that, at any point during this traversal, each edge $(u,v)$ either contributes $w(u,v)$ or $-w(u,v)$ to the score of every edge on the $u$-to-$v$ path in $T$, depending on whether the tour is yet to descend below LCA$(u,v)$ or has already done so.
As above, each update can be implemented in $O(\log n)$ time. Since every edge appears in exactly one list, the total time to perform all the updates is $O(m \log n)$. 

\begin{lemma}
Consider the point in time when the Euler tour had just encountered an edge $e$ for the first time. At that time, for every descendant edge $e'$ of $e$, the weight of the cut determined by $\{e,e'\}$ is $A(e)$ plus the score of $e'$.
\end{lemma}
 
 \begin{proof} Observe that the weight of the cut determined by $\{e,e'\}$ is the sum of weights of all edges with (1) one endpoint in $T_e - T_{e'}$ and the other not in $T_e$, or (2) one endpoint in $T_e - T_{e'}$ and the other in $T_{e'}$.
Note that the edges satisfying (1) have exactly one endpoint in $T_e$, and hence their weight is accounted for in $A(e)$. However, $A(e)$ also counts the weight of edges $(u,v)$ with one endpoint in $T_{e'}$ and the other not in $T_e$. Such edges do not cross the cut. Note that for such edges both $e$ and $e'$ are on the $u$-to-$v$ path in $T$. The fact that $e$ is on the $u$-to-$v$ path implies that the traversal has already descended below LCA$(u,v)$. Hence, $(u,v)$ currently contributes $-w(u,v)$ to the score of $e'$, offseting its contribution to $A(e)$.
Next note that the edges satisfying (2) are edges $(u,v)$ with both $u$ and $v$ in $T_e$, which means that they are not accounted for in $A(e)$, and that the traversal did not yet descend below LCA$(u,v)$. Hence the contribution of such an edge $(u,v)$ to the score of $e'$ is indeed the weight of $(u,v)$.
\end{proof}

By the lemma, the descendant edge $e'$ of $e$ that minimizes the weight of the cut determined by $\{e,e'\}$ is the edge with minimum score in the subtree of $e$ at that time. The score of this edge $e'$ can be found in $O(\log n)$ time by calling $T.\textsc{subtree}(x)$, where $x$ is the lower endpoint of $e$.

\subsection{Independent edges}
We now describe an $O(m\log n)$ time algorithm for finding the minimum cut determined by all pairs of edges $\{e,e'\}$ where $e$ is {\em independent} of $e'$ in $T$ (i.e. $e$ is not a descendant of $e'$ and $e'$ is not a descendant of $e$). We begin by showing that the problem can be reduced to the following {\em bipartite} problem: 

\begin{definition}[The bipartite problem] \label{def:bipartite}
Given two trees $T_1$ and $T_2$ with costs on the edges and a list of non-tree edges $L=\{(u,v) :  u\in T_1, v\in T_2  \}$ where each non-tree edge has a cost, find a pair of edges $e\in T_1$ and $e'\in T_2$ that minimize the sum of costs of $e$, of $e'$, and of all non-tree edges $(u,v)\in L$ where $u$ is in ${T_1}_e$, and $v$ is in ${T_2}_{e'}$. The size of such a problem is defined as the number of non-tree edges in $L$ plus the sizes of $T_1$ and $T_2$. 
\end{definition}

\begin{lemma}\label{lem:reduction}
Given an edge-weighted graph $G$ with $n$ vertices and $m$ edges and a spanning tree $T$, finding the minimum cut among those determined by a pair of independent edges $\{e,e'\}$ can be reduced in $O(m\log n)$ time to multiple instances of the bipartite problem of total size $O(m)$.
\end{lemma} 

\begin{proof}
Recall that every node $w$ of $T$ has at most two children. We create a separate bipartite problem for every node $w$ of $T$ that has exactly two children ($x$ and $y$). This bipartite problem will be responsible for finding the minimum cut determined by all pairs of independent edges $\{e,e'\}$ where $e$ is in $T_x$ and $e'$ is in $T_y$. 

Throughout our description, note the distinction between edge weights and edge costs. The input graph $G$ has edge weights, and the goal is to find the cut with minimum weight. The bipartite problems we define have edge costs, which are derived from the weights of edges in the input graph. 

We initialize the cost of every edge of $G$ to be zero. Then, for every edge $f=(u,v)$ of $G$, we add the weight of $f$ to the cost of every edge on the $u$-to-$v$ path. We maintain the costs in a link-cut tree so each $f$ is handled in $O(\log n)$ time. Now consider any node $w$ with exactly two children $x$ and $y$, and any pair of independent edges $\{e,e'\}$ where $e$ is in $T_x$ and $e'$ is in $T_y$. 
Observe that the edges crossing the cut determined by $\{e,e'\}$ are exactly the edges $f=(u,v)$ with one endpoint in $T_e$ or in $T_{e'}$, and the other endpoint not in $T_e$ nor in $T_{e'}$. Hence, the weight of the cut determined by $\{e,e'\}$ equals the sum of the cost of $e$ plus the cost of $e'$ minus twice the total weight of all non-tree edges $f=(u,v)$ such that $u$ is in $T_e$ and $v$ is in $T_{e'}$. 

We therefore define the bipartite problem for $w$ as follows: (1) $T_1$ is composed of the edge $(w,x)$ and the subtree rooted at $x$ with costs as described above, (2) $T_2$ is composed of the edge $(w,y)$ and the subtree rooted at $y$ with the costs as described above, and (3) for every non-tree edge $f=(u,v)$ with weight $c$ such that LCA$(u,v)=w$ the list of non-tree edges $L$ includes $(u,v)$ with cost $-2c$. 
By construction, the solution to this bipartite problem is the pair of independent edges $e,e'$ with $e \in T_x$ and $e' \in T_y$ that minimize the weight of the cut in $G$ defined by $e$ and $e'$. 

The only issue with the above bipartite problem is that the overall size of all bipartite problems (over all nodes $w$) might not be $O(m)$. This is because the edges of $T$ might appear in the bipartite problems defined for more than a single node $w$. In order to guarantee that the overall size of all bipartite problems is $O(m)$, we construct a compact bipartite problem using the topologically induced trees of Lemma~\ref{lem:topology}.

We construct in $O(m)$ time a constant-time LCA data structure~\cite{HT1984} for $T$.
In overall $O(m\log n)$ time, we construct, for each node $w\in T$ with exactly two children $x$ and $y$:
\begin{enumerate}
	\item A list $L_w$ of all non-tree edges $(u,v)$ with LCA$(u,v) = w$.
	\item A list $\Lambda_x = \{w,x\} \cup \{u : (u,v) \in L_w \text{ and } u \in T_x\}$, sorted according to their visit time in a preorder traversal of T.
	\item A list $\Lambda_y = \{w,y\} \cup \{v : (u,v) \in L_w \text{ and } v \in T_y\}$, sorted according to their visit time in a preorder traversal of T.
\end{enumerate}
These lists require $O(m)$ space and can be easily computed in $O(m\log n)$ time by going over the non-tree edges, because each non-tree edge is in the list $L_w$ of a unique node $w$.

The list $L$ for the compact bipartite problem is identical to the list $L$ for the non-compact problem.
The tree $T^\circ_1$ ($T^\circ_2$) for the compact bipartite problem of $w$ is the topologically tree induced on $(w,x) \cup T_x$ ($(w,y) \cup T_y$) by $\Lambda_x$  ($\Lambda_y$). This is done in $O(|\Lambda_x| \log n)$ ($O(|\Lambda_{y}|\log n)$) time by invoking Lemma~\ref{lem:topology}.
It follows that the total time for constructing all compact bipartite problems is $O(m\log n)$ and the their total space is $O(m)$.

It remains to argue that the solution to the compact bipartite problem is identical to the solution to the non-compact one.
Observe that the cost of a solution $e,e'$ for the non-compact bipartite problem is the cost of $e$ plus the cost of $e'$ plus the cost of all edges in $L$ with one endpoint in ${T_1}_e \cap \Lambda$ and the other endpoint in ${T_2}_{e'} \cap \Lambda$. 

Consider now any pair of edges $e$ and $e'$ for the non-compact bipartite problem. By definition of topologically induced trees, the minimum cost edge $f \in T^\circ_1$ with ${T^\circ_1}_f \cap \Lambda = {T_1}_e \cap \Lambda$, has cost not exceeding that of $e$. An analogous argument holds for $e'$ and an edge $f'$ of $T^\circ_2$. Hence, the cost of the optimal solution for the compact problem is not greater than that of the non-compact problem. Conversely, for any edge $f$ in $T^\circ_1$, there exists an edge $e$ in $T_1$ with cost not exceeding that of $f$ and ${T^\circ_1}_f \cap \Lambda = {T_1}_e \cap \Lambda$. Hence, the cost of the optimal solution for the compact problem is not less than that of the non-compact problem. It follows that the two solutions are the same.
\end{proof} 

The proof of Theorem~\ref{thm:2respecting} follows from the above reduction and the following solution to the bipartite problem:

\begin{lemma}
A bipartite problem of size $m$ can be solved in $O(m\log m)$ time.
\end{lemma} 

\begin{proof}

Recall that in the bipartite problem we are given two trees $T_1$ and $T_2$ with edge-costs and a list of non-tree edges $L=\{(u,v) :  u\in T_1, v\in T_2\}$ where each non-tree edge has a cost.
To prove the lemma, we describe a recursive $O(m\log m)$ time algorithm that finds, for every edge $e\in T_1$, the best edge $e'\in T_2$ (i.e. the edge $e'$ that minimizes the sum of costs of $e'$ and of all non-tree edges $(u,v)\in L$ where $u$ in ${T_1}_e$ and $v \in {T_2}_{e'}$). 

We begin by applying a standard heavy path decomposition~\cite{HT1984} to $T_1$, guided by the number of non-tree edges: The {\em heavy} edge of a node of $T_1$ is the edge leading to the child whose subtree has the largest number of incident non-tree edges in $L$ (breaking ties arbitrarily). The other edges are called {\em light}. The maximal sets of connected heavy edges define a decomposition of the nodes into {\em heavy paths}.

We define a {\em fragment} of the tree $T_1$ to be a subtree of $T_1$ formed by a contiguous subpath $u_1-u_2-\cdots-u_k$ of some heavy path of $T_1$, together with all subtrees hanging from this subpath via light edges. 
Given a fragment $f$, let $L(f) = \{(x_1,y_1), (x_2,y_2) \dots, (x_\ell,y_\ell)\}$ be the set of edges $(x,y)$ of $L$ with $x \in f$. We define the induced subtree $T_2(f)$ to be the tree topologically induced on $T_2$ by the root of $T_2$ and $\{y_1, y_2, \dots, y_{\ell}\}$. The size of $T_2(f)$ is $|T_2(f)| = O(|L(f)|)$. 
We also define a modified induced subtree $T'_2(f)$ as follows.
Let $L_\downarrow(f)$ be the set of edges $(x,y)\in L $ with $x$ in the subtree rooted at the heavy child of the last node $u_k$ of the fragment $f$ (if such a heavy child exists).   
Consider the tree $T_2$, where the cost of each edge $e'$ of $T_2$ is increased by the total cost of all edges $(x,y)\in L_\downarrow(f)$, where $y$ is in ${T_2}_{e'}$.
The modified induced subtree $T'_2(f)$ is defined as the tree topologically induced by the root of $T_2$ and $\{y_1, y_2, \dots, y_\ell\}$  on this modified $T_2$. 

\begin{figure}[h]
\begin{center}
\includegraphics[scale=0.98]{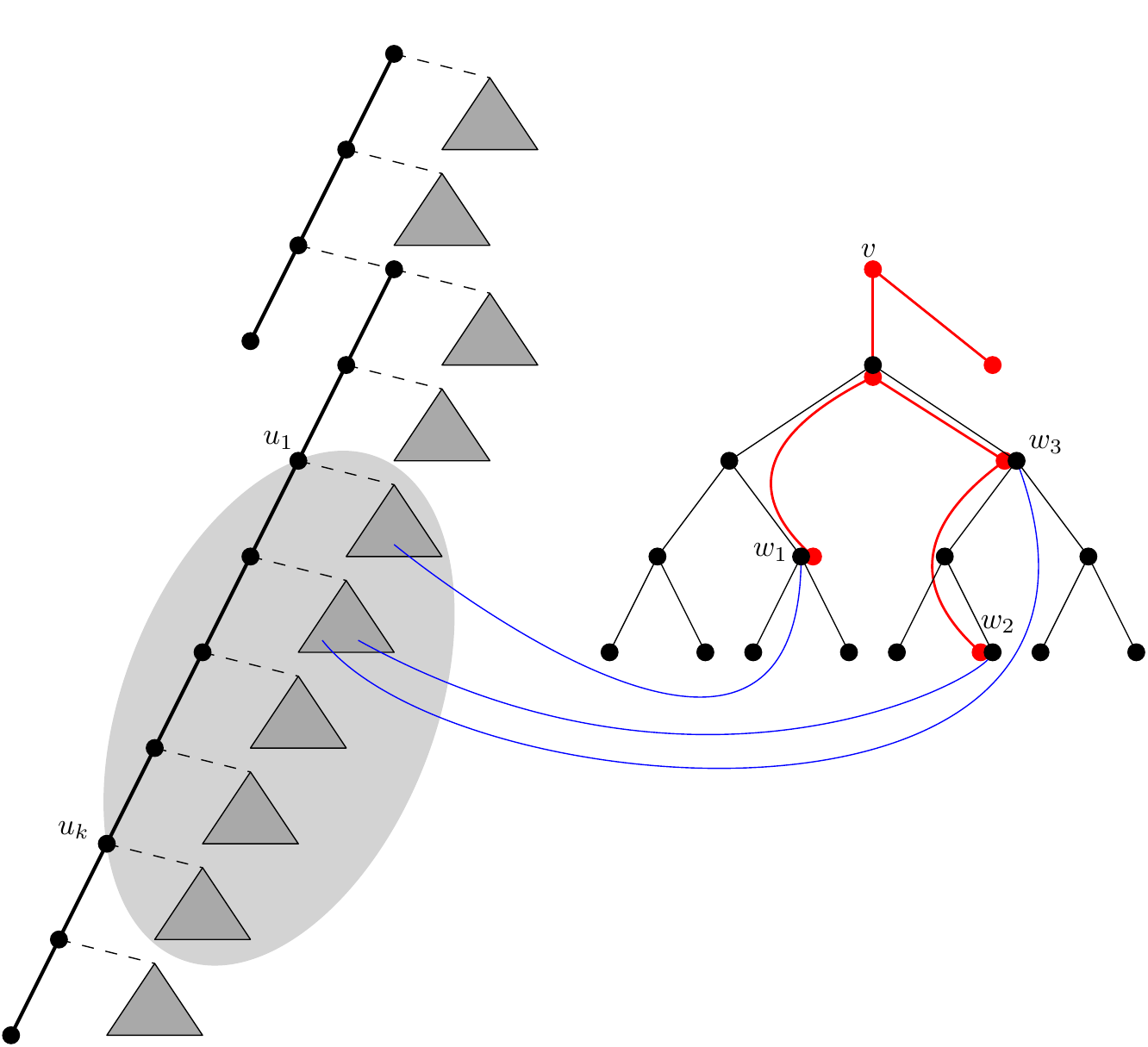}
\end{center}
\caption{On the left: A fragment (in light gray) in the tree $T_1$,  defined by the top node $u_{1}$ and the bottom node $u_{k}$, both laying on the same heavy path (solid edges). The triangles (in dark gray) are the subtrees hanging from the heavy paths via light edges (dashed). On the right: The tree $T_2$ (black) connected to the fragment via three non-tree edges (blue). The endpoints $w_1,w_2,w_3$ of these edges define the topologically induced tree (in red).}
\label{fig:fragment}
\end{figure}

We are now ready to describe the recursion. The input to a recursive call is a fragment $f$ of $T_1$ and the list $(x_1,y_1),(x_2,y_2),\ldots,(x_\ell,y_\ell)$ of all non-tree edges in $L$ with $x_i$ in $f$, together with $T_{2}(f)$ and $T'_{2}(f)$. 
A fragment $f$ is specified by the {\em top} node ($u_1$) and the {\em bottom} node ($u_k$)  of the  corresponding subpath of a heavy path of $T_1$. In the first call, $f$ is specified by the root of $T_1$ and the leaf ending the heavy path of $T_1$ that contains the root. That is, in the first call $f$ is the entire tree $T_1$. The list of non-tree edges for the first call is the entire list $L$. 
  The recursion works by selecting the {\em middle} node of the subpath, defined as follows: We define the {\em light size} $s_i$ of node $u_i$ as the number of non-tree edges $(x,y)\in L$ where either $x=u_i$ or $x$ is in the subtree rooted at the light child of $u_i$. Note that $s_1+s_2+\ldots+s_k = |L(f)|$.  
If $s_1 > |L(f)|/2$ then the middle node is defined as $u_1$. Otherwise, the middle node is defined as the node $u_i$ such that $s_1+\ldots+s_{i-1} \le  |L(f)|/2$ but $s_1+\ldots+s_i > |L(f)|/2$. We keep, for every heavy path $P$ of $T_1$ a list of the nodes of $P$ with non-zero light size, ordered according to their order on $P$. We find the middle node $u_i$ in $O(|L(f)|)$ time by going over the nodes in this list one after the other until we encounter the middle node.
 
After identifying the middle node $u_i$ we apply recursion on the following three fragments: the fragment defined by subpath $u_1-\cdots-u_{i-1}$, the fragment defined by subpath $u_{i+1}-\cdots-u_k$, and the fragment consisting of the entire subtree rooted at the light child of $u_i$.
Before a recursive call to fragment $g$, we construct the appropriate $T_2(g)$ and $T'_2(g)$. The induced tree $T_2(g)$ can be computed from $T_2(f)$ in $O(|T_2(f)|) = O(|L(f)|)$ time by invoking Lemma~\ref{lem:topology} on $T_2(f)$ with $\Lambda_g = r \cup \{y : (x,y)\in L(g)\}$, where $r$ is the root of $T_2$. Note that we had defined $T_2(g)$ as the topologically induced tree on $T_2$ by $\Lambda_g$, not on $T_2(f)$ by $\Lambda_g$. However, since $\Lambda_g \subseteq \Lambda_f$, by Proposition~\ref{prop:topo-compose}, the two definitions are equivalent.

For constructing $T'_2(g)$ from $T'_2(f)$, we first need to increase the cost of each edge $\tilde e$ of $T'_2(f)$ by the total cost of edges in $L_\downarrow(g) \setminus L_\downarrow(f)$ that are incident to $T'_2(f)_{\tilde e}$. 
This can be done in a single bottom up traversal of $T'_2(f)$ in $O(|T'_2(f)|)$ time. Then, we invoke Lemma~\ref{lem:topology} on $T'_2(f)$ with $\Lambda_g$ to obtain $T'_2(g)$. 
To summarize, constructing the trees $T_2(g)$ and $T'_2(g)$ for all three recursive subproblems takes $O(|L(f)|)$ time.

The three recursive calls will take care of finding the best edge $e'\in T_2$ for every edge $e \in T_1$ included in one of the recursive problems. It only remains to handle the three edges that do not belong to any of the recursive problems; the edge between $u_i$ and its light child, the edge $(u_{i-1},u_i)$, and the edge $(u_i,u_{i+1})$. For each such edge $e$ we describe a procedure that finds its best $e'\in T_2(f)$ in time $O(|T_2(f)|)$. 

Recall that, by definition of the bipartite problem, the best edge $e'$ for $e$ is the edge $e'$ of $T_2$ minimizing the cost of $e'$ plus the cost of all non-tree edges $(x,y)\in L$ with $x \in {T_1}_e$, and $y \in {T_2}_{e'}$. 
For the case where $e$ is the edge between $u_i$ and its light child, 
${T_1}_e = T_1(f)_e$. We therefore
mark all non-tree edges $(x,y)\in L(f)$, where $x$ is in $T_1(f)_e$. 
A non-efficient solution would work directly on $T_2$ by propagating, in a bottom up traversal of $T_2$, the cost of all marked edges so that, after the propagation, the cost of every edge $e'$ in $T_2$ has been increased by the total cost of all non-tree edges $(x,y) \in L$ with $x \in T_1(f)_e$, and with $y \in {T_2}_{e'}$. Then we can take the edge $e' \in T_2$ with minimum cost. However, this would take $O(|T_2|) = O(m)$, which is too slow.
Instead, we perform the propagation in $T_2(f)$.  Namely, in a bottom up traversal of $T_2(f)$, we propagate the cost of all marked edges so that, after the propagation, the cost of every edge $e'$ in $T_2(f)$ has been increased by the total cost of all non-tree edges $(x,y) \in L(f)$ with $x \in T_1(f)_e$, and with $y \in T_2(f)_{e'}$. This takes $O(|T_2(f)|) = O(|L(f)|)$ time. Since the propagation process affects all the edges $\tilde e$ with the same ${T_2}_{\tilde e} \cap \Lambda_f$ in the same way, the definition of topologically induced tree guarantees that the edges with minimum cost in $T_2$ and in $T_2(f)$ have the same cost, so using $T_2(f)$ instead of $T_2$ is correct.

The procedure for the cases where $e$ is the edge $(u_{i-1},u_i)$ or $(u_i,u_{i+1})$ is identical, except that we apply it with $T'_2(f)$ instead of $T_2(f)$. This difference stems from the fact that applying the above procedure on $T_2(f)$ only considers the costs of the non-tree edges in $L(f)$, but not the costs of the non-tree edges in $L_\downarrow(f)$, which might also cross cuts involving the edges $(u_{i-1},u_i)$ or $(u_i,u_{i+1})$. The definition of the costs of edges in $T'_2(f)$ takes into account the contribution of costs of non-tree edges in $L_\downarrow(f)$.
The rest of the propagation procedure and the proof of its correctness remain unchanged.

To analyze the overall running time, let $T(m)$ be the time to handle a fragment $f$ corresponding to a whole heavy path,
and $T'(m)$ be the time to handle a fragment $f$ corresponding to a proper subpath of some heavy path, where $m=|L(f)|$. Then
$T(m)=T'(m_{1})+T(m_{2})+T'(m_{3})$ for some $m_{1},m_{2},m_{3}$, where $m_{1}+m_{2}+m_{3}=m$ since the subproblems are disjoint,
$m_{1},m_{3}\leq m/2$ by the choice of the middle node, and $m_{2}\leq m/2$ by the definition of a heavy path. This is since the light child of $u_i$ does not have more incident non-tree edges in its subtree than the number of non-tree edges incident to the subtree of the heavy child $u_{i+1}$.  When $f$ corresponds to a whole heavy path the number of non-tree edges incident to the subtree of $u_{i+1}$ is exactly $m_3$.
Similarly, for the case where $f$ does not correspond to a whole heavy path, $T'(m) = T'(m_{1})+T(m_{2})+T'(m_{3})$ for some $m_{1},m_{2},m_{3}$, where $m_{1}+m_{2}+m_{3}=m$ and $m_{1},m_{3}\leq m/2$
(but now we cannot guarantee that $m_{2}\leq m/2$). Considering the tree describing the recursive calls, on any path from the root (corresponding
to the fragment consisting of the whole $T_{1}$) to a leaf, we have the property that the value of $m$ decreases by at least a factor of $2$ every two steps.
Hence, the depth of the recursion is $O(\log m)$. It follows that the total time to handle a bipartite problem of size $m$ is $O(m \log m)$.
\end{proof}

\section{A $\log\log n$ Speedup}

Karger modified his $O(m\log^{3}n)$ time minimum cut algorithm to work in
$O(m\log^{3}n/\log\log n+n\log^{6}n)$ time by observing that 1-respecting cuts can
be found faster than 2-respecting cuts, and tweaking the parameters of the tree packing.
In this section we explain how to apply this idea, together with our new $O(m\log n)$
time algorithm for the 2-respecting problem, to derive a new $O(m\log^{2}n/\log\log n+n\log^{3+\delta}n)$
time minimum cut algorithm, for any $\delta>0$.

As in Karger's implementation~\cite[Section 9.1]{Karger}, we start with an initial sampling step, except that instead of $\epsilon=\frac{1}{4\log n}$ we use
$\epsilon=\frac{1}{4\log^{\gamma}n}$, for some $\gamma\in (0,1)$ to be fixed later. This produces in linear time a graph $H$
with $m'=O(n/\epsilon^{2}\log n)$ edges and minimum cut $c'=O(\epsilon^{-2}\log n)$ such that the minimum cut in $G$ corresponds to a $(1+\epsilon)$-times minimum cut in $H$. We find a packing in $H$ of weight $c'/2$ with Gabow's algorithm~\cite{Gabow}
in $O(m'c'\log n)=O(n/\epsilon^{4}\log^{3}n)$ time.
We choose $4\log^{1+\gamma}n$ trees at random from the packing and for each tree we find the minimum cut
that 1-respects it. This takes total $O(m\log^{1+\gamma}n)$ time.
Then, we choose $\log n/\log(\log^{\gamma}n)=O(\log n/\log\log n)$ trees at random
from the packing and for each tree we find the minimum cut that 2-respects it. This takes total in $O(m\log^{2}n/\log\log n)$ time using our new algorithm.

Let $\rho \geq c'/2$ be the weight of the packing, let $\alpha\rho$ be the total weight of trees
that 1-respect the minimum cut,  and let $\beta\rho$ be the total weight of trees that 2-respect the minimum cut.
As observed by Karger, $\beta \geq 1-2\epsilon-2\alpha$.
Thus, either $\alpha > \frac{1}{4\log^{\gamma}n}$ and choosing $4\log^{1+\gamma}n$ trees
guarantees that none  of them 1-respects the minimum cut with probability at most
\[  (1-\alpha)^{4\log^{1+\gamma}n} \leq \left(1-\frac{1}{4\log^{\gamma}n}\right)^{4\log^{\gamma}n\cdot \log n} < 1/n, \]
or $\beta \geq 1-1/\log^{\gamma}n$ and choosing $\log n/\log(\log^{\gamma}n)$
trees guarantees that none of them 2-respects the minimum cut with probability
at most 
\[ (1-\beta)^{\log n/\log(\log^{\gamma}n)} \leq \left(\frac{1}{\log^{\gamma}n}\right)^{\log n/\log(\log^{\gamma}n)} = 1/n. \]

\noindent 
The overall complexity is $O(n/\epsilon^{4}\log^{3}n+m\log^{1+\gamma}n+m\log^{2}n/\log\log n)=
O(m\log^{2}n/\log\log n+n\log^{3+4\gamma}n)$. By adjusting $\gamma=\delta/4$ we obtain
the claimed complexity of $O(m\log^{2}n/\log\log n+n\log^{3+\delta}n)$.

\section*{Acknowledgements}
We thank Daniel Anderson and Guy Blelloch for drawing our attention to an inaccuracy in a prior version of Section~\ref{sec:descendant}.

\bibliographystyle{abbrv}

\appendix
\section{A constant factor approximation of the minimum cut}
\label{app:Matula}
Matula~\cite{Matula} gave an $O(m/\epsilon)$ time algorithm that finds a $(2+\epsilon)$ approximation of the minimum cut in an unweighted graph $G$.
The algorithm proceeds in iterations, where each iteration takes $O(m)$ time, and either finds a $(2+\epsilon)$ approximate cut, or produces a subgraph $G'$ that contains the minimum cut of $G$, but has only a constant fraction of the edges of $G$.  Hence there are $O(\log n)$ iterations, and the total running time is $O(m)$ for any fixed $\epsilon$. 

Matula's algorithm can be easily extended to the weighted setting. Each iteration can be implemented in $O(m\log n)$ time (cf.~\cite{KargerThesis}), and produces a subgraph $G'$ with a constant factor of the total edge weight of $G$. Thus, the algorithm produces a $(2+\epsilon)$ approximation of the minimum cut in a weighted graph $G$ in $O(m \log n \log W)$ time, where $W$ is the sum of edge weights in $G$.  

The running time can be decreased to $O(m \log^2 n)$ at the expense of a worse constant factor approximation as follows.
Let $G$ be a weighted graph. Let $c$ denote the weight of the minimum cut in $G$. 
Compute a maximum spanning tree $T$ of $G$, and let $w^*$ be the minimum weight of an edge in $T$. It is easy to see~\cite{Karger99} that $w^* \leq c \leq n^2 w^*$. 
Contract all edges $e$ with $w(e) > n^2 w^*$. Clearly, this does not affect the minimum cut.
If $w^* \leq n^3$, then all edge weights are now bounded by $n^5$, and we can run Matula's algorithm in $O(m \log^2 n)$ time, and obtain a $(2+\epsilon)$-approximate minimum cut.
Otherwise, set $\tilde w(e) \leftarrow \lfloor w(e)/\frac{w^*}{n^3} \rfloor$, and delete all edges with $\tilde w = 0$. Call the resulting graph $\tilde G$.
Observe that $\tilde G$ has integer edge weights bounded by $n^5$, so we can find a $(2+\epsilon)$-approximate minimum cut in $\tilde G$ in $O(m \log^2 n)$ time.
Now scale up the edge weights in $\tilde G$ by $\frac{w^*}{n^3}$. The weight of each edge in $\tilde G$ is now off from its original weight in $G$ by at most $\frac{w^*}{n^3} \leq \frac{c}{n^3}$. Thus, the weights of any cut in $\tilde G$ and in $G$ differ by at most $c/n$. Hence, an $(2+\epsilon)$-approximate minimum cut in $\tilde G$ is an $O(1)$-approximate minimum cut in $G$.        

\end{document}